\newif\iflong
\longtrue

\newif\ifold\oldfalse

\documentclass[a4paper,UKenglish,cleveref,autoref, thm-restate]{lipics-v2021}

\iflong
\pdfoutput=1 
\nolinenumbers
\hideLIPIcs 

\else

\fi

\newcommand{\G}{\ensuremath{\mathcal{G}}}
\newcommand{\prob}[3]{
\begin{definition}[\textsc{#1}]\ \\
	Input: #2\\
	Question: #3
\end{definition}
}

\renewcommand{\prob}[3]{\begin{quote}  \textsc{#1}\\  \textbf{Input:} #2\\  \textbf{Question:} #3\end{quote}}

\title{Distance to Transitivity: New Parameters for Taming Reachability in Temporal Graphs} 

\titlerunning{Distance to Transitivity: Parameters for Taming Reachability in Temporal Graphs.} 

\author{Arnaud Casteigts}{Department of Computer Science, University of Geneva, Switzerland}{arnaud.casteigts@unige.ch}{https://orcid.org/0000-0002-7819-7013}	{supported by French ANR, project ANR-22-CE48-0001 (TEMPOGRAL).}

\author{Nils Morawietz}{Friedrich Schiller University Jena, Institute of Computer Science, Germany}{nils.morawietz@uni-jena.de}{https://orcid.org/0000-0002-7283-4982}{}

\author{Petra Wolf}{LaBRI, CNRS, Université de Bordeaux, Bordeaux INP, France \and 
		\url{https://www.wolfp.net/}}
	{mail@wolfp.net}
	{https://orcid.org/0000-0003-3097-3906}
	{supported by French ANR, project ANR-22-CE48-0001 (TEMPOGRAL).}

\authorrunning{A. Casteigts, N. Morawietz, and P. Wolf}

\Copyright{Arnaud Casteigts, Nils Morawietz, and Petra Wolf} 



\category{} 

\relatedversion{}

\usepackage{amssymb}
\usepackage{graphicx}
\usepackage{xspace}
\usepackage{centernot}
\usepackage{tikz}
\usetikzlibrary{calc}

\usetikzlibrary{positioning,chains,fit,shapes,calc}

\pgfdeclarelayer{background}
\pgfdeclarelayer{foreground}
\pgfsetlayers{background,main,foreground}

\newcommand{\tvd}{\ensuremath{\delta_{\mathrm{vd}}}\xspace}
\newcommand{\taa}{\ensuremath{\delta_{\mathrm{aa}}}\xspace}
\newcommand{\FPT}{\textsf{FPT}\xspace}

\newcommand{\NP}{\textsf{NP}\xspace}

\newcommand{\tae}{\ensuremath{\delta_{\mathrm{am}}}\xspace}

\newcommand{\coNPpoly}{\textsf{coNP/poly}\xspace}

\newcommand{\W}{{\textsf{{W[1]}}}\xspace}

\newcommand{\bth}{$\NP\subseteq\coNPpoly$}

\newcommand{\Oh}{\mathcal{O}}

\usepackage[disable]{todonotes}

\newcommand{\CL}{\textsc{Clique}\xspace}
\newcommand{\OTCC}{\textsc{Open-TCC}\xspace}
\newcommand{\CTCC}{\textsc{Closed-TCC}\xspace}

\newcommand{\iin}[1][v]{{#1}_{\mathrm{in}}}

\newcommand{\mg}{\mathcal{G}}
\newcommand{\wg}{\widetilde{\mg}}

\newcommand{\vc}{\mathrm{vc}}

\newcommand{\omin}{\mathrm{out}^{\min}}
\newcommand{\imax}{\mathrm{in}^{\max}}

\begin{document}
\maketitle
\begin{abstract}
  A temporal graph is a graph whose edges only appear at certain
  points in time. Reachability in these graphs is defined in terms of
  paths that traverse the edges in chronological order (temporal
  paths). This form of reachability is neither symmetric nor
  transitive, the latter having important consequences on the
  computational complexity of even basic questions, such as computing
  temporal connected components. In this paper, we introduce several
  parameters that capture how far a temporal graph~$\mathcal{G}$ is
  from being transitive, namely, \emph{vertex-deletion distance to
    transitivity} and \emph{arc-modification distance to transitivity}, both
  being applied to the reachability graph of $\mathcal{G}$. We
  illustrate the impact of these parameters on the temporal connected
  component problem, obtaining several tractability results in terms
  of fixed-parameter tractability and polynomial kernels.
  Significantly, these results are obtained without restrictions of
  the underlying graph, the snapshots, or the lifetime of the input
  graph. As such, our results isolate the impact of non-transitivity
  and confirm the key role that it plays in the hardness of
  temporal graph problems.
\end{abstract}

\newpage
\section{Introduction}

Temporal graphs have gained attention lately as appropriate tools to
capture time-dependent phenomena in fields as various as
transportation, social networks analysis, biology, robotics,
scheduling, and distributed computing. On the theoretical side, these
graphs generate interest mostly for their intriguing features. Indeed,
many basic questions are still open, with a general feeling that
existing techniques from graph theory typically fail on temporal
graphs. In fact, most of the natural questions considered in static
graphs turn out to be intractable when formulated in a temporal
version, and likewise, most of the temporal analogs of classical
structural properties are false.

One of the earliest examples is that the natural analog of Menger's
theorem does not hold in temporal graphs~\cite{KKK00}. Another early
result is that deciding if a temporal connected component (set of vertices that
can reach each other through temporal paths) of a certain size exists
is NP-complete~\cite{BF03}. A more recent and striking result is that
there exist temporally connected graphs on $\Theta(n^2)$ edges in
which every edge is critical for connectivity; in other words, no
temporal analog of sparse spanners exist unconditionally~\cite{AF16}
(though they do, probabilistically~\cite{CMMZ21}). Moreover, minimizing
the size of such spanners is APX-hard~\cite{AGMS17,AF16}. Further
hardness results for problems whose static versions are generally
tractable include separators~\cite{FMN+20}, connectivity
mitigation~\cite{EMMZ19}, exploration~\cite{ES22,AFGW23},
flows~\cite{ACG+19}, Eulerian paths~\cite{BM23}, and even spanning
trees~\cite{CC24}.

Faced by these difficulties, the algorithmic community has focused on
special cases, and tools from parameterized complexity were employed
with moderate success. A natural approach here is to apply the range
of classical graph parameters to restrict either the underlying graph
of the temporal graph (i.e. which edges can exist at all) or its
snapshots (i.e. which edges may exist simultaneously). For example,
finding temporal paths with bounded waiting time at each node (which is
NP-hard in general) turns out to be FPT when parameterized by
treedepth or vertex cover number of the underlying graph. But the
problem is already W[1]-hard for pathwidth (let alone
treewidth)~\cite{CHMZ21}. 
In fact, as observed in~\cite{treewidth}, most
temporal graphs problems remain hard even when the underlying graph
has bounded treewidth (sometimes, even a tree or a
star~\cite{EMMZ19,AMSR21,AFGW23}).

A possible explanation for these results is that temporal graph
problems are \emph{very} hard. Another one is that parameters based on
static graph properties are not adequate. Some
parameters whose definition is based on that of a temporal
graph include timed feedback vertex sets (counting the cumulative
distance to trees over all snapshots)~\cite{CHMZ21} and the
$p(\mathcal{G})$ parameter from~\cite{AFGW23}, that measures in a
certain way how dynamic the temporal graph is and enables polynomial
kernels for the exploration problem. While these parameters
represent some progress towards finer-grained restrictions,
they remain somewhat structural in
the sense that their definition is stable under re-shuffling of the
snapshots.

A key aspect of temporal graphs is that the \emph{ordering of events}
matters. Arguably, a truly temporal parameter should be sensitive to
that. An interesting step in this direction was recently made by
Bumpus and Meeks~\cite{BM23}, introducing interval-membership-width, a
parameter that quantifies the extent to which the set of intervals
defined by the first and last appearance of an edge at each vertex can
overlap (with application to Eulerian paths). In a sense, this
parameter measures how complex the interleaving of events could be.
Another, perhaps even more fundamental feature of temporal graphs is
that the reachability relation based on temporal paths is not guaranteed to be
symmetric or transitive. While the former is a well-known limitation
of directed graphs, the latter is specific to temporal graphs
(directed or not), and it has been suspected to be one of the main
sources of intractability since the onset of the theory. (Note that a
temporal graph of bounded interval-membership-width may still be
arbitrarily non-transitive.) In the present work, we explore new
parameters that control how transitive a temporal graph is, thereby
isolating, and confirming, the role that this aspect plays in the
tractability of temporal reachability problems.

\subparagraph{Our Contributions.}
We introduce and investigate two parameters that measure how far a
temporal graph is from having transitive reachability. For a temporal
graph \G, our parameters directly address the reachability features of
$\G$, and as such, they are formulated in terms of its reachability
graph $G_R=(V,\{(u,v) : u \leadsto v\})$, a directed graph whose arcs
represent the existence of temporal paths in $\G$, whether $\G$ itself
is directed or undirected. Indeed, the reachability of $\G$ is
transitive if and only if the arc relation of~$G_R$ is transitive. Two natural ways of
measuring this distance are in terms of \emph{vertex deletion} and
\emph{arc modification}, namely:
\begin{itemize}
\item \emph{Vertex-deletion distance to transitivity} (\tvd) is the minimum number of vertices
  whose deletion from $G_R$ makes the resulting graph transitive.
\item \emph{Arc-modification distance to transitivity} (\tae) is the minimum number
of arcs whose addition or deletion from $G_R$ makes the resulting
graph transitive.
\end{itemize}
Observe that vertex addition does not appear to make particular sense.
As for the arc-modification distance, we may occasionally consider its
restriction to arc-addition only (\taa).

Among the many problems that were shown intractable in temporal
graphs, one of the first, and perhaps most iconic one, is the
computation of temporal connected components~\cite{BF03} (see
also~\cite{CLMS23,SSBP24}). In order to benchmark our new parameters,
we investigate their impact on the computational complexity of this
problem. Informally, given a temporal graph $\G$ (defined later) on a
set of vertices $V$, a temporal connected component is a subset
$V' \subseteq V$ such that for all $u$ and $v$ in $V'$, $u$ can reach
$v$ by a temporal path (shorthand, $u \leadsto v$). Interestingly, the
non-transitive nature of reachability here makes it possible for such
vertices to reach each other through temporal paths that travel
outside the component, without absorbing the intermediate vertices into
the component. This gives rise to two distinct notions of components:
\emph{open} temporal connected components (\OTCC) and \emph{closed}
temporal connected components (\CTCC), the latter requiring that only
internal vertices are used in the temporal paths, and both being
NP-hard to compute.

The statement of our results requires a few more facts. Both
algorithmic and structural results in temporal graphs are highly
sensitive to subtle definitional variations, called \emph{settings}.
In the \emph{non-strict} setting, the labels along a temporal path are
only required to be non-decreasing, whereas in the \emph{strict}
setting, they must be increasing. It turns out that both settings are
sometimes incomparable in difficulty, and the techniques developed for
each may be different. Some temporal graphs, called \emph{proper},
have the property that no two adjacent edges share a common time
label, making it possible to ignore the distinction between strict and
non-strict temporal paths. Whenever possible, hardness results should
preferably be obtained for proper temporal graphs, so that they
apply in both settings at once. Finally, with a few exceptions, our
results hold for both directed and undirected temporal graphs.

Bearing these notions in mind, our results are the following. 
For
\OTCC, we obtain an FPT algorithm with parameter \tvd, running in time
$3^{\tvd}\cdot n^{\Oh(1)}$ (in all the settings). 
Unfortunately, \tvd
turns out to be too small for obtaining a polynomial size kernel. 
In
fact, we show that under reasonable computational complexity
assumptions, no polynomial kernel in $\tvd + \vc + k$ exists (except
possibly for the non-strict undirected setting), where~$k$ denotes the size of the sought tcc and where~$\vc$ denotes the vertex cover number of the underlying graph. 
Next, we obtain an
FPT algorithm running in time $4^{\tae}\cdot n^{\Oh(1)}$ for the
mostly larger parameter \tae, and show that \OTCC admits a polynomial kernel
of size $|M|^3$, where~$M$ is a given arc set for which~$(V,A(G_R) \Delta M)$ is transitive. 
It also admits a polynomial kernel of size $\taa^2$ when
restricting modification to addition-only (again, all these results hold in
all the settings). 
\CTCC, in comparison, seems to be a harder problem,
at least with respect to our parameters. 
In particular, we show that
it remains NP-hard even if $\tae=\tvd=1$ in all the settings (through
proper graphs). 
It is also W[1]-hard when parameterized by~$\tvd+\tae+k$ in all the settings,
except possibly in the non-strict undirected setting. In fact, these
two results hold even for temporal graphs whose reachability graph
misses a single arc towards being a bidirectional clique.

Put together, these results establish clearly that non-transitivity is
a genuine source of hardness for \OTCC. The case of \CTCC is less
clear. On the one hand, the parameters do not suffice to make this
particular version of the problem tractable. This is not so
surprising, as the reachability graph itself does not encode which
paths are responsible for reachability, in particular, whether these
paths are internal or external in a component. On the other hand, this
gives us a separation between both versions of the problem and
provides some support for the fact that \CTCC may be harder than
\OTCC, which was not known before. Finally, the negative results for
\CTCC can serve as a landmark result for guiding future efforts in
defining transitivity parameters that exploit more sophisticated
structures than the reachability graph.

\subparagraph{Organization of the Work.}
The main definitions are given in~\cref{sec:preliminaries}. Then, we investigate each parameter in a dedicated section (\tvd in~\cref{sec:tvd} and \tae in~\cref{sec:tae}). The limitations of these parameters in the case of \CTCC are presented in~\cref{sec:limits}. Finally, \cref{sec:conclusion} concludes the paper with some remarks and open questions.
\iflong\else Due to space limitations, the proofs of statements marked with~$(\star)$ are deferred to Appendix.\fi

\section{Preliminaries}
\label{sec:preliminaries}
For concepts of parameterized complexity, like FPT, \W-hardness, and polynomial kernels, we refer to the standard monographs~\cite{DF13,C+15}. A reduction $g$ between two parameterized problems is called a \emph{polynomial parameter transformation}, if the reduction can be computed in polynomial time and, if for every input instance $(I, k)$, we have that $(I', k') = g(I,k)$ with~$k' \in k^{\mathcal{O}(1)}$.
We call a polynomial time reduction from a problem $P$ to $P$ itself a \emph{self-reduction}.

\medskip
\noindent
\textbf{Notation.}
Let $n$ be a positive integer, we denote with $[n]$ the set $\{1, 2, \dots, n\}$. For a decision problem $P$, we say that two instances $I_1, I_2$ of~$P$ are~\emph{equivalent} if $I_1$ is a yes-instance of $P$ if and only if $I_2$ is a yes-instance of $P$. For two sets $A$ and $B$, we denote with $A \Delta B$ the symmetric difference of $A$ and $B$.

\medskip
\noindent
\textbf{Graphs.}
We consider a graph $G = (V, E)$ to be a static graph. If not indicated otherwise, we assume $G$ to be undirected.
Given a (directed) graph~$G$, we denote by $V(G)$ the set of vertices of $G$, by  $E(G)$ (respectively, $A(G)$) the set of edges (arcs) of $G$.
Let $G = (V, E)$ be a graph a let $X \subseteq V(G)$ be a set of vertices.
We denote by~$E_G(X) = \{\{u,v\}\in E\mid u\in X, v\in X\}$ the edges in~$G$ between the vertices of~$S$.
Moreover, we define the following operations on~$G$:
$G[X] = (X, E_G[X])$, $G - X = G[V \setminus X]$. 
We call a sequence $\rho = v_0, v_1, \ldots, v_r$ of vertices  a \emph{path} in graph~$G$ if $v_0,\ldots,v_r\in V(G)$ and for $i \in [r]$, $\{v_{i-1}, v_i\} \in E(G)$. 
We denote with $N_G[v]$ the closed neighborhood of the vertex $v\in V(G)$.
A vertex set~$S\subseteq V$ is a~\emph{clique} in an undirected graph, if each pair of vertices in~$S$ is adjacent in~$G$. 
For a directed graph $G = (V, A)$, we call a set $S \subseteq V$ a \emph{bidirectional clique}, if for every pair of distinct vertices $u, v$ in $S$, we have $(u,v) \in A$ and $(v, u) \in A$.
Let~$G = (V, A)$ be a directed graph.
A~\emph{strongly connected component (scc)} in~$G$ is an inclusion maximal vertex set~$S\subseteq V$ under the property that there is a directed path in~$G$ between any two vertices of~$S$.
For each directed graph~$G$, there is a unique partition of the vertex set of~$G$ into sccs.
Moreover, this partition can be computed in polynomial time~\cite{T72}.

\medskip
\noindent
\textbf{Temporal graphs.}
A \emph{temporal graph} $\mathcal{G}$ over a set of vertices $V$ is a sequence $\mathcal{G} = (G_1, G_2, \dots, G_L)$ of graphs such that for all $t \in [L], V(G_t) = V$. We call $L$ the \emph{lifetime} of $\mathcal{G}$ and for $t \in [L]$, we call $G_t = (V, E_t)$ the \emph{snapshot graph} of $\mathcal{G}$ at \emph{time step} $t$. 
We might refer to $G_t$ as $\mathcal{G}(t)$. 
We call $G = (V, E)$ with $E = \bigcup_{t \in [L]} E_t$ the \emph{underlying graph} of $\mathcal{G}$.
We denote by $V(\mathcal{G})$ the set of vertices of $\mathcal{G}$. We write $V$ if the graph or temporal graph
is clear from context.
We call an undirected temporal graph $\mathcal{G} = (G_1, G_2, \dots, G_L)$ \emph{proper}, if for each vertex $v\in V(\mathcal{G})$ the degree of $v$ in $G_t$ is one, for each $t \leq L$. We call a directed temporal graph $\mathcal{G} = (G_1, G_2, \dots, G_L)$ \emph{proper}, if for each vertex $v\in V(\mathcal{G})$ either the out-degree or the in-degree of $v$ in $G_t$ is zero, for each $t \leq L$.
We further call a (directed) temporal graph $\mathcal{G}$ \emph{simple}, if each edge (arc) exists in exactly one snapshot.
We call a sequence of vertices $v_0, v_1, \ldots, v_r$ that form a path in the underlying graph~$G$ of~$\mathcal{G}$ a \emph{strict (non-strict) temporal path} in $\mathcal{G}$ if for each $i \in [r]$, there exists an $j_i \in [L]$ such that $\{v_{i-1}, v_i\} \in E(G_{j_i})$ and the sequence of indices $j_i$ is increasing (non-decreasing).

For a temporal graph $\mathcal{G}$, we say that a vertex $u \in V$ \emph{strictly (non-strictly) reaches} a vertex $v \in V$ if there is a strict (non-strict) temporal path from $u$ to $v$, i.e., with $v_0 = u$ and $v_r = v$.
We define the \emph{strict (non-strict) reachability relation} $R \subseteq V \times V$ as: for all $u, v \in V$, $(u, v) \in R$ if and only if $u$ strictly (non-strictly) reaches $v$. We call the directed graph $G_R = (V, R)$ the \emph{strict (non-strict) reachability graph} of $\mathcal{G}$. We say that $G_R$ is transitive, resp. symmetric, if and only if $R$ is. 
More generally, we say that a directed graph $G$ is \emph{transitive}, if its set of arcs forms a transitive relation.
For a directed graph $G = (V, A)$ we call a set of vertices $S \subseteq V$ a \emph{transitivity modulator} if $G - S$ is transitive.

\begin{observation}
	\label{lem:trans-closed-delete}
	Let $G$ be a transitive directed graph. Then, for any $v \in V(G)$, $G[V \setminus \{v\}]$ is also transitive.
\end{observation}

Next we define our main problems of interest in this work, finding open and closed temporal connected components.
\prob{\textsc{Open Temporal Connected Component (\OTCC)}}
{Temporal graph $\mathcal{G} = (G_1, G_2, \dots, G_L)$ and integer $k$.}
{Does there exists an open temporal connected component of size at least $k$, i.e., a subset $C\subseteq V(\mathcal{G})$ with $|C| \geq k$, such that for each $u, v \in C$, $u$ reaches $v$, and vice versa.}
We differentiate between the strict vs.~non-strict and directed vs.~undirected version of \OTCC depending on whether we consider strict vs.~non-strict reachability and directed vs.~undirected temporal graphs.
We define the problem \textsc{Closed Temporal Connected Component (\CTCC)} similarly with the additional restriction that at least one path over which $u$ reaches $v$ is fully contained in $C$. 
We abbreviate a temporal connected component as tcc.

\medskip
\noindent
\textbf{Distance to transitivity.} We introduce two parameters that measure how far the reachability graph $G_R=(V,A)$ of a temporal graph is from being transitive. 
The first parameter, \emph{vertex-deletion distance to transitivity}, \tvd, counts how many vertices need to be deleted from $G_R$ in order to obtain a transitive reachability graph, i.e., the size of a minimum transitivity modulator. 
This parameter is especially suited for temporal graphs for which the reachability graph consists of cliques with small overlaps. 
The second parameter, \emph{arc-modification distance to transitivity}, \tae, counts how many arcs need to be added to or removed from $G_R$ in order to obtain a transitive reachability graph and is especially suited for directed temporal graphs or temporal graphs for which the reachability graph consists of cliques with large overlaps.
Formally, we define the parameters as follows. 
\begin{align*}
	\tvd &= \min_{S \subseteq V}(|S|)\text{ for which }G_R' = G_R-S\text{ is transitive.}\\
	\tae &= \min_{M \subseteq V \times V}(|M|)\text{ for which  }G_R' = (V, A \Delta M)\text{ is transitive.}	
\end{align*}
For $\tae$, we call the set $M$ an arc-modification set.
Note that $\tvd \leq 2\cdot \tae$, since the endpoints of an arc-modification set form a transitivity modulator.

\subsection{Basic Observations}
Next, we present basic observations that motivate the study of the considered parameters.
\begin{lemma}[\cite{BF03}]
	\label{lem:tcc-are-cliques}
	Let $\mathcal{G}$ be a temporal graph with reachability graph $G_R$. Then a set $S\subseteq V(\mathcal{G})$ is a tcc in $\mathcal{G}$ if and only if $S$ is a bidirectional clique in $G_R$.
\end{lemma}
\iflong
\begin{proof}
	Assume $S$ is a tcc in $\mathcal{G}$. Then for each pair of distinct vertices $u, v \in V(\mathcal{G})$ we have that $u$ reaches $v$ and $v$ reaches $u$. As $G_R$ represents the reachability relation of $\mathcal{G}$, we hence have $(u, v) \in A(G_R)$ and $(v, u) \in A(G_R)$. Hence, $S$ forms a bidirectional clique in $R_G$.
	
	For the other direction, assume $S$ is a bidirectional clique in $G_R$. Hence, for each pair of distinct vertices $u, v\in V(G_R)$, we have that $(u, v) \in A(G_R)$ and $(v, u) \in A(G_R)$. As $R_G$ represents the reachability relation of $\mathcal{G}$, we have that $u$ reaches $v$ and $v$ reaches $u$ in $\mathcal{G}$. Hence, $S$ forms a temporal connected component in $\mathcal{G}$.
\end{proof}
\fi
\begin{lemma}
	\label{lem:scc-are-cliques-old}
	Let $G$ be a transitive directed graph. 
	Then every vertex set~$S\subseteq V(G)$ is a bidirectional clique in~$G$ if and only if each pair of vertices of~$S$ can reach each other.
	\end{lemma}
\iflong
\begin{proof}
	First, assume $S \subseteq V(G)$ is a bidirectional clique in $G$. Then, for each pair of vertices $u, v \in V(G)$ we have that $(u, v), (v,u) \in A(G)$ and, hence, $u$ and $v$ can reach each other.
	
	Next, assume $S \subseteq V(G)$ is not a bidirectional clique in $G$ and let $u, v$ be two vertices in $S$ for which $(u,v) \notin A(G)$. For the sake of contradiction, assume that each pair of vertices in $S$ can reach each other. Then, there exists a path $w_1, w_2, \dots w_\ell$ such that $w_1 = u$ and $w_\ell = v$. By the assumption that $G$ is transitive, this implies that $(u,v) \in G$; a contradiction.
\end{proof}
\fi
Note that this implies the following.

\begin{corollary}
	\label{lem:scc-are-cliques}
	Let $G$ be a transitive directed graph. 
	Then every scc in $G$ is also a maximal bidirectional clique and vice versa.  
\end{corollary}

The previous observations thereby imply that both \OTCC and \CTCC can be solved in polynomial time on temporal graphs with transitive reachability graphs.

\section{Vertex-Deletion Distance to Transitivity}
\label{sec:tvd}
We first focus on the parameter $\tvd$. 
Note that computing this parameter is \NP-hard: In a strict temporal graph~$\mg$ with lifetime~1, the reachability graph~$G_R$ of~$\mg$ is exactly the directed graph obtained from orienting each edge of the underlying graph in both directions.
Hence, on such a temporal graph, computing~$\tvd$ is exactly the cluster vertex deletion number of the underlying graph, that is, the minimum size of any vertex set to remove, such that no induced path of length 2 remains.
Since computing the latter parameter is \NP-hard~\cite{DBLP:journals/jcss/LewisY80}, this hardness also translates to computing the parameter $\tvd$.

Moreover, note that computing this parameter can be done similarly to computing the cluster vertex deletion number of a graph:
If a directed graph~$G=(V,A)$ is not transitive, then there are vertices~$u,v,$ and~$w$ in~$V$, such that~$(u,v)$ and~$(v,w)$ are arcs of~$A$ and~$(u,w)$ is not an arc of~$A$.
Hence, each transitivity modulator for~$G$ has to contain at least one of the vertices~$u,v,$ or~$w$.
This implies, that a standard branching algorithm that considers each of these three vertices to be removed from the graph to obtain a transitive graph, finds a minimum size transitivity modulator in $3^{\tvd} \cdot n^{\Oh(1)}$~time. 
\begin{proposition}\label{compute tvd}
	Let $\mathcal{G}$ be a temporal graph with reachability graph $G_R$. 
	Then we can compute a minimal set $S \subseteq V(\mathcal{G})$ of size $\tvd$ such that $G_R[V \setminus S]$ is transitive in time $3^{\tvd}\cdot n^{\Oh(1)}$.
\end{proposition}

Based on this result, we now present an \FPT-algorithm for~\textsc{Open TCC} when parameterized by~$\tvd$.

\begin{lemma}\label{tvd algo if set known}
	Let~$I:=(\mg,k)$ be an instance of~\OTCC  with reachability graph $G_R$. 
	Let $S$ be a given transitivity modulator of $G_R$. 
	Then, we can solve~$I$ in time $2^{|S|} \cdot n^{\Oh(1)}$.
\end{lemma}
\begin{proof}	
By Lemma~\ref{lem:scc-are-cliques}, every scc in $G_R[V\setminus S]$ is a bidirectional clique, since~$S$ is a transitivity modulator for~$G_R$.
Lemma~\ref{lem:tcc-are-cliques} then implies that each tcc~$C$ in $\mathcal{G}$ with~$C\cap S = \emptyset$ is an scc in~$G_R[V\setminus S]$ and vice versa.

	The \FPT-algorithm then works as follows: 
	We iterate over all subsets $S'$ of $S$ with the idea to find a tcc that extends $S'$. 
	If $S'$ is not a bidirectional clique in~$G_R$, we discard the current set and continue with the next subset of~$S$, as no superset of $S'$ is a bidirectional clique and thus also not a tcc.
	Hence, assume that $S'$ is a bidirectional clique.
If~$S'$ has size at least~$k$, $I$ is a trivial yes-instance of~\OTCC.
Otherwise, we do the following:
Let~$V'$ be the vertices of $V \setminus S$ that are bidirectional connected to every vertex in $S'$. 
As $G_R[V \setminus S]$ is transitive, \Cref{lem:trans-closed-delete} implies that $G_R[V']$ is also transitive. 
Hence, the sccs in $G_R[V']$ correspond to tccs in $\mathcal{G}$ by~\Cref{lem:scc-are-cliques} and \Cref{lem:tcc-are-cliques}.
	Since every vertex in $S'$ is bidirectional connected to every other vertex in $S' \cup V'$ in $G_R$, for each bidirectional clique~$C \subseteq V '$ in~$G_R[V']$, $C \cup S'$ is a tcc in~$\mg$.
	Hence, it remains to check, whether any scc in $G_R[V']$ has size at least $k- |S'|$. 
	\Cref{tvd figure if set known} illustrates the sets $S$, $S'$, and $V'$.
	
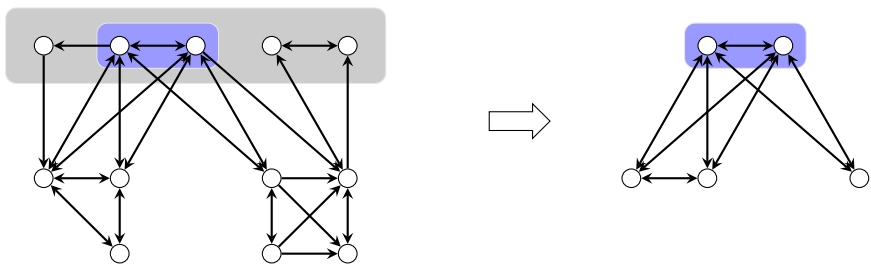
\begin{figure}
\begin{center}
\begin{tikzpicture}

\tikzstyle{knoten}=[circle,fill=white,draw=black,minimum height=6pt,minimum width=7pt,inner sep=0pt]
\tikzstyle{bez}=[inner sep=0pt]

\newcommand\drawbgrect[3][blue!50]{
\node (bl) at (#2) {};
\node (tr) at (#3) {};

\draw[rounded corners, fill=#1, draw=black!10] ($(bl) + (-.4,-.5)$) rectangle ($(tr) + (.4,.5)$) {};
}

\begin{scope}[yshift=-50]

\node[knoten] (v1) at (0,0) {};
\node[knoten] (v2) at ($(v1) + (1,0)$) {};
\node[knoten] (v4) at ($(v2) + (0,-1)$) {};

\node[knoten] (v3) at ($(v2) + (2,0)$) {};
\node[knoten] (v5) at ($(v4) + (2,0)$) {};
\node[knoten] (w1) at ($(v3) + (1,0)$) {};
\node[knoten] (w2) at ($(v5) + (1,0)$) {};

\

\draw[stealth-stealth,thick] (v1) to (v2);
\draw[stealth-stealth,thick] (v2) to (v4);
\draw[stealth-stealth,thick] (v4) to (v1);

\draw[stealth-stealth,thick] (v3) to (v5);
\draw[stealth-stealth,thick] (w1) to (w2);
\draw[-stealth,thick] (v3) to (w1);
\draw[-stealth,thick] (v5) to (w1);
\draw[-stealth,thick] (v3) to (w2);
\draw[-stealth,thick] (v5) to (w2);

\end{scope}

\node[knoten] (x1) at (0,0) {};

\foreach \x [count=\xi] in {2,3,4,5}{        
 \node[knoten] (x\x) at ($(x\xi) + (1,0)$) {};
}

  \tikzstyle{every path}=[black]
  
\draw[-stealth,thick] (x1) to (v1);
\draw[-stealth,thick] (x2) to (x1);
\draw[stealth-stealth,thick] (x4) to (x5);
\draw[stealth-stealth,thick] (w1) to (x4);
\draw[-stealth,thick] (w1) to (x5);

\draw[-stealth, thick] (x3) to (w1);

\draw[stealth-stealth,thick] (x2) to (x3);

\draw[stealth-stealth, thick] (x2) to (v1);
\draw[stealth-stealth, thick] (x3) to (v1);
\draw[stealth-stealth, thick] (x2) to (v2);
\draw[stealth-stealth, thick] (x3) to (v2);

\draw[stealth-stealth, thick] (x2) to (v3);
\draw[stealth-stealth, thick] (x3) to (v3);

\begin{pgfonlayer}{background}

\draw[rounded corners, fill=black!20, draw=black!10] ($(x1) + (-.5,-.5)$) rectangle ($(x5) + (.5,.5)$) {};

\draw[rounded corners, fill=blue!40, draw=black!10] ($(x2) + (-.3,-.3)$) rectangle ($(x3) + (.3,.3)$) {};

\end{pgfonlayer}

    \node[single arrow, draw=black,
      minimum width = 8pt, single arrow head extend=3pt,
      minimum height=8mm] at ($(x5) + (2.2,-1)$) {};

\begin{scope}[xshift=220]

\begin{scope}[yshift=-50]

\node[knoten] (v1) at (0,0) {};
\node[knoten] (v2) at ($(v1) + (1,0)$) {};

\node[knoten] (v3) at ($(v2) + (2,0)$) {};

\draw[stealth-stealth,thick] (v1) to (v2);

\end{scope}

\node[knoten] (x2) at (1,0) {};
\node[knoten] (x3) at (2,0) {};

\draw[stealth-stealth,thick] (x2) to (x3);

\draw[stealth-stealth, thick] (x2) to (v1);
\draw[stealth-stealth, thick] (x3) to (v1);
\draw[stealth-stealth, thick] (x2) to (v2);
\draw[stealth-stealth, thick] (x3) to (v2);

\draw[stealth-stealth, thick] (x2) to (v3);
\draw[stealth-stealth, thick] (x3) to (v3);

\begin{pgfonlayer}{background}

\draw[rounded corners, fill=blue!40, draw=black!10] ($(x2) + (-.3,-.3)$) rectangle ($(x3) + (.3,.3)$) {};

\end{pgfonlayer}

\end{scope}

\end{tikzpicture}
\end{center}

\caption{Illustration of the algorithm in \Cref{tvd algo if set known}. On the left: reachability graph $G_R$ with transitivity modulator $S$ in gray and the chosen subset~$S' \subseteq S$ to extend in blue. On the right: The subset $S'$ together with the vertices $V'$ that are bidirectionally connected to all vertices in $S'$.}
\label{tvd figure if set known}
\end{figure}
	
	Finding the strongly connected components of a graph and identifying whether a set of vertices forms a bidirectional clique can be done in polynomial time. 
	Hence, our algorithm runs in time~$2^{\tvd}\cdot n^{\Oh(1)}$, since we iterate over each subset~$S'$ of~$S$. 
\end{proof}

Based on~\Cref{compute tvd} and~\Cref{tvd algo if set known}, we thus derive our \FPT-algorithm for~\OTCC when parameterized by~$\tvd$.

\begin{theorem}\label{algo tvd}
	\OTCC can be solved in $3^{\tvd} \cdot n^{\Oh(1)}$~time.
\end{theorem}

\subsection*{Kernelization Lower Bounds}
In this section, we show that a polynomial kernel for~\OTCC when parameterized by~$\tvd + \vc$ is unlikely, where~$\vc$ is the vertex cover number of the underlying graph.
Note that~$\tvd$ and~$\vc$ are incomparable:
On the one hand, consider a temporal graph~$\mg$ where the underlying graph~$G$ is a star with leaf set~$X\cup Y$ and center~$c$, such that the edges from~$X$ to~$c$ exist in snapshots~$G_1$ and~$G_3$ and the edges from~$Y$ to~$c$ exist in snapshot~$G_2$.
Then, each vertex of~$X$ can reach each other vertex, but in the strict setting, no vertex of~$Y$ can reach any other vertex of~$Y$.
Hence, each minimum transitivity modulator has to contain all vertices of~$X$ or all but one vertex of~$Y$, which implies that for~$|X|=|Y|$, $\tvd \in \Theta(|V(\mg))$, whereas the vertex cover number of~$G$ is only~$1$.
On the other hand, consider a temporal graph~$\mg$ with only one snapshot~$G_1$, such that~$G_1$ is a clique.
Then, the underlying graph of~$\mg$ is exactly~$G_1$ and has a vertex cover number of~$|V(\mg)|-1$, but the strict reachability graph of~$\mg$ is a bidirectional clique, which is a transitive graph.
Hence, $\tvd(\mg) = 0$.

We now present our kernelization lower bound for the strict undirected version of~\OTCC.

\begin{theorem}
	The strict undirected version of~\OTCC does not admit a polynomial kernel when parameterized by~$\vc + \tvd + k$, unless~\bth, where~$\vc$ denotes the vertex cover number of the underlying graph.
\end{theorem}
\begin{proof}
	This result immediately follows from the known~\cite{KKK00} reduction from~\CL which, in fact, is as a polynomial parameter transformation.
	
\prob{\CL}
{An undirected graph $G = (V,E)$ and integer $k$.}
{Is there a clique of size~$k$ in~$G$?}

	For the sake of completeness, we recall the reduction.	
	Let~$I:=(G=(V,E),k)$ be an instance of~\CL and let~$\mg$ be the temporal graph with lifetime 1, where~$G$ is the unique snapshot of~$\mg$.
	
	Then, for each vertex set~$X\subseteq V$, $X$ is a clique in~$G$ if and only if~$X$ is a strict tcc in~$\mg$.
	Hence, $I$ is a yes-instance of~\CL if and only if~$(\mathcal{G},k)$ is a yes-instance of the strict undirected version of~\OTCC.
	It is known that~\CL does not admit a polynomial kernel when parameterized by~$k$ plus the vertex cover number of~$G$, unless~\bth~\cite{C+05}.
	Let~$S$ be a minimum size vertex cover of~$G$ and let~$G_R$ be the strict reachability graph of~$\mg$.
	Note that~$G_R$ contains an arc~$(u,v)$ with~$u\neq v$ if and only if~$\{u,v\}$ is an edge of~$G$.
	Hence, $V\setminus S$ is an independent set in~$G_R$, which implies that~$S$ is a transitivity modulator of~$G_R$.
	Consequently, $\tvd \leq |S|$.
	Recall that~\CL does not admit a polynomial kernel when parameterized by~$k + |S|$, unless~\bth~\cite{C+05}.
	This implies that the strict and undirected version of~\OTCC does not admit a polynomial kernels when parameterized by the vertex cover number of the underlying graph plus~$\tvd$ plus~$k$, unless~\bth.
\end{proof}

Next, we present the same lower bound for both directed versions of~\OTCC.

\begin{theorem}\iflong\else[$\star$]\fi\label{no kernel directed}
	The directed version of~\OTCC does not admit a polynomial kernel when parameterized by~$\vc + \tvd + k$, unless~\bth, where~$\vc$ denotes the vertex cover number of the underlying graph.
	This holds both for the strict and the non-strict version of the problem.
\end{theorem}

\newcommand{\proofNoKern}{}
\begin{proof}
	Again, we present a polynomial parameter transformation from~\CL.

	Recall that~\CL does not admit a polynomial kernel when parameterized by the size of a give minimum size vertex cover~$S$ of~$G$ plus~$k$, unless~\bth~\cite{C+05}.
	This holds even if~$G[S]$ is~$(k-1)$-partite~\cite{GK20}, which implies that each clique of size~$k$ in~$G$ contains exactly~$k-1$ vertices of~$S$ and exactly one vertex of~$V\setminus S$, since~$V \setminus S$ is an independent set.

	\subparagraph{Construction.}
	Let~$I:=(G:=(V,E),k)$ be an instance of~\CL and let~$S$ be a given minimum size vertex cover~$S$ of~$G$, such that~$G[S]$ is~$(k-1)$-partite.
	Assume that~$k > 6$.
	
	We obtain an equivalent instance of~\OTCC in two steps:
	First, we perform an adaptation of a known reduction~\cite{BF03} from the instance~$(G[S],k-1)$ of~\CL to an instance~$(\widetilde{\mathcal{G}},k-1)$ of the directed version of~\OTCC where each sufficiently large (of size at least~$5$) vertex set~$X$ of~$\widetilde{\mathcal{G}}$ is a tcc in~$\widetilde{\mathcal{G}}$ if and only if~$X$ is a clique in~$G[S]$.
	Second, we extend~$\widetilde{\mathcal{G}}$ by the vertices of~$V\setminus S$ and some additional connectivity-gadgets, to ensure that the resulting temporal graph has a tcc of size~$k$ if and only if there is a vertex from~$V\setminus S$ for which the neighborhood in~$G$ contains a clique of size~$k-1$.

\begin{figure}[t]
  \centering
    \begin{tikzpicture}
      \tikzstyle{every node}=[draw,circle,fill=darkgray,inner sep=2pt]
      \path (-.2,1) node (vv){};
      \path (-2.2,1) node (uu){};

\begin{scope}[xshift = 20]
      \path (4.65,1) node (v){};
      \path (2.65,1) node (u){};
      \path (3.65,1.4) node (v'){};
      \path (3.65,.6) node (u'){};
\end{scope}

      \tikzstyle{every node}=[inner sep=4pt]

    \node[single arrow, draw=black,
      minimum width = 8pt, single arrow head extend=3pt,
      minimum height=8mm] at (1.5,1) {};

      \path (u) node[left] {$u$};
      \path (v) node[right] {$v$};
      \path (uu) node[left] {$u$};
      \path (vv) node[right] {$v$};
      \path (u') node[below] {$e_{uv}$};
      \path (v') node[above] {$e_{vu}$};
      \draw[thick] (uu)--(vv);
      \draw[-stealth, thick] (u)-- node[below]{4}(u');
      \draw[-stealth, thick] (u')-- node[below]{5}(v);
      \draw[-stealth, thick] (v)-- node[above]{4}(v');
      \draw[-stealth, thick] (v')-- node[above]{5}(u);
    \end{tikzpicture}
    \caption{\label{fig:semaphore}For two adjacent vertices~$u$ and~$v$ of~$S$ the vertices and arcs added to the temporal graph~$\wg$ in the proof of~\Cref{no kernel directed}.}
  \end{figure}
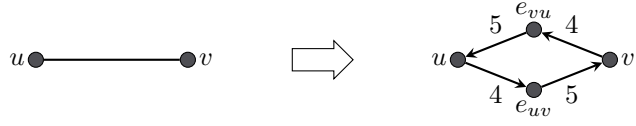
	
	Let~$(\widetilde{\mathcal{G}},k-1)$ be the instance of~\OTCC constructed as follows:
We initialize~$\wg$ as an edgeless temporal graph of lifetime~$5$ with vertex set~$S \cup \{e_{uv}, e_{vu}\mid \{u,v\}\in E_G(S)\}$.	
Next, for each edge~$\{u,v\}\in E$, we add the arcs~$(u, e_{uv})$ and~$(v, e_{vu})$ to time step 4 and add the arcs~$(e_{uv}, v)$ and~$(e_{vu}, u)$ to time step 5. 
This completes the construction of~$\wg$.
An example of the arcs added to~$\wg$ is shown in~\Cref{fig:semaphore}.
Note that the first three snapshots of~$\wg$ are edgeless.
	This construction is an adaptation of the reduction presented by Bhadra~and~Ferreira~\cite{BF03} to the case of directed temporal graphs.
	Note that the temporal graph~$\wg$ has the following properties that we make use of in our reduction:
	\begin{enumerate}[1)]
		\item $\wg$ is a proper and simple directed temporal graph,
		\item the vertex set~$\mathcal{V}$ of~$\wg$ has size~$\Oh(|S|^2)$ and contains all vertices of~$S$,
		\item each tcc of size at least~$k-1$ in~$\wg$ contains only vertices of~$S$, and
		\item each vertex set~$X\subseteq S$ of size at least~$k-1$ is a tcc in~$\wg$ if and only if~$X$ is a clique in~$G[S]$.
	\end{enumerate}
	Note that the two last properties imply that the largest tcc of~$\wg$ has size at most~$k-1$, since~$G[S]$ is~$(k-1)$-partite.
	
	Next, we describe how to extend the temporal graph~$\wg$ to obtain a temporal graph~$\mg'$ which has a tcc of size~$k$ if and only if~$I$ is a yes-instance of~\CL.
	Let~$n := |V|$.
	Moreover, let~$\mg'$ be a copy of~$\wg$.
	We extend the vertex set of~$\mg'$ by all vertices of~$V\setminus S$, and a vertex~$\iin$ for each vertex~$v\in S$.
	
	For each vertex~$v\in S$, we add the arc~$(\iin,v)$ to time step~$3$.
	For each vertex~$v\in S$ and each neighbor~$w\in V\setminus S$ of~$v$ in~$G$, we add the arc~$(v,w)$ to time step~$2$ and the arc~$(w,\iin)$ to time step~$1$.
	This completes the construction of~$\mg'$.
	Let~$V'$ denote the newly added vertices, that is, $V' := (V\setminus S) \cup \{\iin\mid v\in S\}$.
	
	Next, we show that there is a clique of size~$k$ in~$G$ if and only if there is a tcc of size~$k$ in~$\mg'$.
	
	$(\Rightarrow)$
	Let~$K\subseteq V$ be a clique of size~$k$ in~$G$.
	We show that~$K$ is a tcc in~$\mg'$.
	As discussed above, $K$ contains exactly~$k-1$ vertices of~$S$ and exactly one vertex~$w^*$ of~$V\setminus S$.
	By construction of~$\wg$, $K\setminus \{w^*\}$ is a tcc in~$\wg$ and thus also a tcc in~$\mg'$.
	It thus remains to show that each vertex~$K\setminus \{w^*\}$ can reach vertex~$w^*$ in~$\mg'$ and vice versa.
	Since each vertex of~$K\setminus \{w^*\}$ is adjacent to~$w^*$ in~$G$, by construction, $w^*$ is an out-neighbor of each vertex of~$K\setminus \{w^*\}$ in~$\mg'$.
	Hence, it remains to show that~$w^*$ can reach each vertex of~$K\setminus \{w^*\}$ in~$\mg'$.
	Let~$v$ be a vertex of~$K\setminus \{w^*\}$.
	Since~$v$ is adjacent to~$w^*$ in~$G$, there is an arc~$(w^*,\iin[v])$ in~$\mg'$ that exists at time step~$1$.
	Hence, there is a temporal path from~$w^*$ to~$v$ in~$\mg'$, since the arc~$(\iin[v],v)$ exists at time step~$3$.
	Concluding, $K$ is a tcc in~$\mg'$.

	$(\Leftarrow)$
	Let~$X$ be a tcc of size~$k$ in~$\mg'$.
	We show that~$X$ is a clique of size~$k$ in~$G$.
	To this end, we first show that~$X$ contains only vertices of~$V$.
	Afterwards, we show that~$X$ is a clique in~$G$.
	
	To show that~$X$ contains only vertices of~$V$, we first analyze the reachability of vertices of~$V(\mg')$.
	For a vertex~$v\in V(\mg')$, we denote 
\begin{itemize}
\item by~$\omin_v$ the smallest time label of any arc exiting~$v$ and
\item by~$\imax_v$ the largest time label of any arc entering~$v$.
\end{itemize}	
	Note that a vertex~$v$ cannot reach a distinct vertex~$w$ in~$\mg'$ if~$\imax_w < \omin_v$.
	\Cref{tab ell values} show for each vertex~$v\in V(\mg')$ a lower bound for $\omin_v$ and an upper bound for~$\imax_v$.
	
	\begin{table}
		\caption{For each vertex~$v\in V(\mg')$ a lower bound for $\omin_v$ and an upper bound for~$\imax_v$.}
		\centering
		\begin{tabular}{l|cc}
			& $\omin_v$ & $\imax_v$ \\\hline
			$v\in S$ & 2 & 5 \\
			$v\in V(\wg)\setminus S$ & 4 & 5 \\
			$v\in V \setminus S$ & 1 & 2 \\
			$v\in \{\iin[u] \mid u\in S\}$ & 3 & 1
		\end{tabular}
		\label{tab ell values}
	\end{table}
	Based on~\Cref{tab ell values}, we can derive the following properties about reachability in~$\mg'$.

	\begin{claim}\label{connectivity stuff}
		\begin{enumerate}[a)]
			\item No vertex of~$V(\wg)\setminus S$ can reach any vertex of~$V'$ in~$\mg'$.\label{con s new}
			\item No vertex of~$\{\iin\mid v\in S\}$ can reach any other vertex of~$\{\iin\mid v\in S\}$ in~$\mg'$.\label{con in in}
			\item No vertex of~$\{\iin\mid v\in S\}$ can reach any vertex of~$V\setminus S$ in~$\mg'$.\label{con in is}
			\item No vertex of~$S$ can reach any vertex of~$\{\iin\mid v\in S\}$ in~$\mg'$.\label{con s in}\item No vertex of~$V\setminus S$ can reach any other vertex of~$V\setminus S$ in~$\mg'$.\label{con is is}		
		\end{enumerate}
	\end{claim}
	\begin{claimproof}
		Based on~\Cref{tab ell values}, we derive Items~\ref{con s new}) to~\ref{con s in}).
		It remains to show~\Cref{con is is}).
		To this end, observe that each arc with a vertex of~$V\setminus S$  as source has a vertex of~$\{\iin\mid v\in S\}$ as sink.
		Due to~\Cref{con in is}, no vertex of~$\{\iin\mid v\in S\}$ can reach any vertex of~$V\setminus S$ in~$\mg'$.
		Hence, no vertex~$V\setminus S$ can reach any other vertex of~$V\setminus S$ in~$\mg'$.
		This implies that~\Cref{con is is}) holds.
	\end{claimproof}
	
	Since~$X$ is a tcc in~$\mg'$, \Cref{connectivity stuff} implies that~$X$ contains at most one vertex of~$V\setminus S$ (due to~\Cref{con is is}) and at most one vertex of~$\{\iin\mid v\in S\}$ (due to~\Cref{con in in}).
	In other words, $X$ contains at most two vertices of~$V'$.
	Since~$k > 6$, this then implies that~$X$ contains at least one vertex of~$V(\wg)$.
	\Cref{connectivity stuff} thus further implies that~$X$ contains no vertex of~$\{\iin\mid v\in S\}$ (due to~\Cref{con s new,con s in}).
	This then implies that~$X$ contains at least~$k-1$ vertices of~$V(\wg)$.
	
	To show that~$X$ contains only vertices of~$V$ and is a clique in~$G$ we now show that the reachability between any two vertices of~$V(\wg)$ in~$\mg'$ is the same as in~$\wg$.
	Let~$P$ be a temporal path between two distinct vertices of~$V(\wg)$ in~$\mg'$.
	We show that~$P$ is also a temporal path in~$\wg$.
	Assume towards a contradiction that this is not the case.
	Hence, $P$ visits at least one vertex of~$V'$.
	Since no vertex of~$V(\wg)$ can reach any vertex of~$\{\iin\mid v\in S\}$ (due to~\Cref{con s new,con s in}), $P$ visits no vertex of~$\{\iin\mid v\in S\}$.
	Moreover, since each vertex of~$V\setminus S$ has only out-neighbors in~$\{\iin\mid v\in S\}$, $P$ visits no vertex of~$V\setminus S$ either.
	Consequently, $P$ contains no vertex of~$V'$; a contradiction.
	
	Hence, $P$ is a temporal path in~$\wg$, which implies that for each vertex set~$Y \subseteq V(\wg)$, $Y$ is a tcc in~$\wg$ if and only if~$Y$ is a tcc in~$\mg'$.
	Recall that~$X$ contains at least~$k-1$ vertices of~$V(\wg)$.
	Since the largest tcc in~$\wg$ has size at most~$k-1$ and each tcc of size~$k-1$ in~$\wg$ is a clique in~$G$, this implies that~$X\cap V(\wg)$ is a clique of size~$k-1$ in~$G[S]$.
	Since~$X$ contains no vertex of~$\{\iin\mid v\in S\}$, this implies that~$X$ contains exactly one vertex~$w^*$ of~$V\setminus S$.
	Hence, it remains to show that each vertex~$v\in X \setminus \{w^*\}$ is adjacent to~$w^*$ in~$G$.
	Since~$X$ is a tcc in~$\mg'$, $v$ can reach~$w^*$ in~$\mg'$.
	By construction and illustrated in~\Cref{tab ell values}, $\omin_v \geq 2 \geq \imax_{w^*}$.
	Since $v$ reaches $w^*$ and $\mg'$ is a proper temporal graph, the arc~$(v,w^*)$ is contained in~$\mg'$.
	By construction, this implies that~$v$ and~$w^*$ are adjacent in~$G$.
	Consequently, $X$ is a clique in~$G$.
	This completes the correctness proof of the reduction.
	
	\subparagraph{Parameter bounds.}
	It thus remains to show that~$\tvd(\mg')$ and the vertex cover of the underlying graph of~$\mg'$ are at most~$|S|^{\Oh(1)}$ each.
	Let~$V^* := V(\mg') \setminus (V\setminus S)$.
	Note that~$V^*$ has size~$|V(\wg)| + |S|  \in \Oh(|S|^2)$ and is a vertex cover of the underlying graph of~$\mg'$.
	Hence, the vertex cover number of the underlying graph of~$\mg'$ is~$\Oh(|S|^2)$.
	To show the parameter bounds, it thus suffices to show that~$V^*$ is a transitivity modulator of the reachability graph~$G_R$ of~$\mg'$.
	Due to~\Cref{connectivity stuff}, $G_R - V^* = G_R[V\setminus S]$ is an independent set.
	Consequently, $V^*$ is a transitivity modulator of~$G_R$.
	Hence, $\tvd(\mg') \in \Oh(|S|^2)$.
	By the fact that~\CL does not admit a polynomial kernel when parameterized by~$|S| + k$, unless~\bth, \OTCC does not admit a polynomial kernel when parameterized by~$\tvd(\mg')$ plus the vertex cover number of the underlying graph of~$\mg'$ plus~$k$, unless~\bth.
\end{proof}
{}
\iflong
\proofNoKern
\fi

Note that our kernelization lower bounds do not include the non-strict undirected version of~\OTCC.
An modification of~\Cref{no kernel directed} seems difficult, unfortunately.
This is due to the fact that undirected edges can be traversed in both direction, which makes it very difficult to limit the possible reachable vertices in the temporal graph, while preserving a small transitivity modulator.

\section{Arc-Modification Distance to Transitivity - A Polynomial Kernel}
\label{sec:tae}
Next, we focus on the parameterized complexity of~\OTCC when parameterized by the size of a given arc-modification set towards a transitive reachability graph.
As discussed earlier, for each arc-modification set~$M$ towards a transitive reachability graph, $\tvd$ does not exceed~$2\cdot |M|$, since removing the endpoints of all edges of~$M$ results in a transitivity modulator.
This implies the following due to~\Cref{algo tvd} and the fact that a minimum size arc-modification set towards a transitive graph can be computed in $2.57^{\tae}\cdot n^{\Oh(1)}$~time~\cite{WKNU12}.
\begin{corollary}
\OTCC can be solved in $4^{\tae} \cdot n^{\Oh(1)}$~time.
\end{corollary}
In the remainder of this section, we thus consider this parameter with respect to kernelization algorithms.
In contrast to parameterizations by~$\tvd$, we now show that a polynomial kernelization algorithm can be obtained for~\OTCC when parameterized by the size of a given arc-modification set towards a transitive reachability graph.

In fact, we show an even stronger result, since our kernelization algorithm does not need to know the actual arc-modification set but only its endpoints.
To formulate this more general result, we need the following definition:
Let~$G=(V,A)$ be a directed graph.
A transitivity modulator~$S\subseteq V$ of~$G$ is called~\emph{inherent}, if there is an arc-modification set~$M$ with~$M\subseteq S\times S$ for which~$(V, A\Delta M)$ is a transitive graph.
Note that the set of endpoint of an arc-modification set towards a transitive graph always forms an inherent transitivity modulator.

\begin{theorem}\label{kernel if set given}
Let~$I=(\mathcal{G},k)$ be an instance of~\OTCC and let~$G_R = (V,A)$ be the reachability graph of~$\mathcal{G}$.
Moreover, let~$B \subseteq V$ be an inherent transitivity modulator of~$G_R$.
Then, for each version of~\OTCC, one can compute in polynomial time an equivalent instance of total size~$\Oh(|B|^3)$.
\end{theorem}

\begin{figure}[t]
  \centering
  \begin{tikzpicture}[yscale=1.2]
  
  \path (4.5,1.5) node (f){};
    \node[single arrow, draw=black,
      minimum width = 8pt, single arrow head extend=3pt,
      minimum height=8mm] at ($(f) + (2.4,-.7)$) {};
  
  \node[] (aaaaaa) at (.6,2) {$\widehat{G}$};
  
\tikzstyle{every node}=[draw,darkgray,fill=blue!50,circle,inner sep=2.5pt]
  \path (.6,1) node (z){};
  \path (1.3,1) node (a){};
  \path (2,1) node (b){};
  \path (3,1) node (d){};
  \path (4,1) node (e){};
  \path (4.5,1.5) node (f){};
  \path (5,1) node (g){};

  \path (1.65,0) coordinate (aa);
  \path (2.9,0) coordinate (bb);
  \path (0.6,0) coordinate (zz);
  \path (4,0) coordinate (dd);
  \path (5,0) coordinate (ee);

  \tikzstyle{every path}=[black]

  \draw[thick] (z)--(a)--(b)--(d)--(e)--(f)--(g)--(e)--(f)--(b);
  
  \tikzstyle{every path}=[gray]
  \draw (aa) -- (a);
  \draw (aa)+(-.3,0) -- (a);
  \draw (aa)+(-.15,0) -- (a);
  \draw (aa)+(.15,0) -- (a);
  \draw (aa)+(.3,0) -- (a);

    \draw (aa) -- (b);
  \draw (aa)+(-.3,0) -- (b);
  \draw (aa)+(-.15,0) -- (b);
  \draw (aa)+(.15,0) -- (b);
  \draw (aa)+(.3,0) -- (b);

      \draw (zz) -- (z);
  \draw (zz)+(-.2,0) -- (z);
  \draw (zz)+(-.1,0) -- (z);
  \draw (zz)+(.1,0) -- (z);
  \draw (zz)+(.2,0) -- (z);

        \draw (dd) -- (d);
  \draw (dd)+(-.25,0) -- (d);
  \draw (dd)+(-.12,0) -- (d);
  \draw (dd)+(.12,0) -- (d);
  \draw (dd)+(.25,0) -- (d);

        \draw (dd) -- (e);
  \draw (dd)+(-.25,0) -- (e);
  \draw (dd)+(-.12,0) -- (e);
  \draw (dd)+(.12,0) -- (e);
  \draw (dd)+(.25,0) -- (e);

        \draw (dd) -- (g);
  \draw (dd)+(-.25,0) -- (g);
  \draw (dd)+(-.12,0) -- (g);
  \draw (dd)+(.12,0) -- (g);
  \draw (dd)+(.25,0) -- (g);

  \tikzstyle{every node}=[draw,darkgray,fill=white,circle,inner sep=8pt]
  \path (aa) node (aaa){};
  \tikzstyle{every node}=[draw,darkgray,fill=white,circle,inner sep=7pt]
  \path (dd) node (ddd){};
  \tikzstyle{every node}=[draw,darkgray,fill=white,circle,inner sep=5pt]
  \path (zz) node (zzz){};
  \tikzstyle{every node}=[draw,darkgray,fill=white,circle,inner sep=4pt]
  \path (bb) node (bbb){};
  \tikzstyle{every node}=[draw,darkgray,fill=white,circle,inner sep=6pt]
  \path (ee) node (eee){};

  \begin{scope}[xshift=220]

  \tikzstyle{every node}=[black]
  
  \node (aaaaaa) at (1.3,2) {$G'$};
  
  \tikzstyle{every node}=[draw,darkgray,fill=blue!50,circle,inner sep=2.5pt]
  \path (1.3,1) node (a){};
  \path (2,1) node (b){};
  \path (3,1) node (d){};
  \path (4,1) node (e){};
  \path (4.5,1.5) node (f){};
  \path (5,1) node (g){};

  \path (1.65,0) coordinate (aa);
  \path (4,0) coordinate (dd);

  \tikzstyle{every path}=[black]

  \draw[thick] (a)--(b)--(d)--(e)--(f)--(g)--(e)--(f)--(b);
  
  \tikzstyle{every path}=[gray]
  \draw (aa) -- (a);
  \draw (aa)+(-.3,0) -- (a);
  \draw (aa)+(-.15,0) -- (a);
  \draw (aa)+(.15,0) -- (a);
  \draw (aa)+(.3,0) -- (a);

    \draw (aa) -- (b);
  \draw (aa)+(-.3,0) -- (b);
  \draw (aa)+(-.15,0) -- (b);
  \draw (aa)+(.15,0) -- (b);
  \draw (aa)+(.3,0) -- (b);


        \draw (dd) -- (d);
  \draw (dd)+(-.25,0) -- (d);
  \draw (dd)+(-.12,0) -- (d);
  \draw (dd)+(.12,0) -- (d);
  \draw (dd)+(.25,0) -- (d);

        \draw (dd) -- (e);
  \draw (dd)+(-.25,0) -- (e);
  \draw (dd)+(-.12,0) -- (e);
  \draw (dd)+(.12,0) -- (e);
  \draw (dd)+(.25,0) -- (e);

        \draw (dd) -- (g);
  \draw (dd)+(-.25,0) -- (g);
  \draw (dd)+(-.12,0) -- (g);
  \draw (dd)+(.12,0) -- (g);
  \draw (dd)+(.25,0) -- (g);

  \tikzstyle{every node}=[draw,darkgray,fill=white,circle,inner sep=7pt]
  \path (aa) node (aaa){};
  \tikzstyle{every node}=[draw,darkgray,fill=white,circle,inner sep=6pt]
  \path (dd) node (ddd){};

\end{scope}

\end{tikzpicture}
\label{fig:while-blue}
\caption{Left: the original instance of~\CL from \Cref{kernel if set given} constructed from the reachability graph of the considered temporal graph.
Right: the obtained compressed instance of~\CL after exhaustive application of all reduction rules.
In both parts, the blue vertices are the vertices from the inherent transitivity modulator~$B$ and the cycles at the bottom indicate the white clusters.
Note that in both graphs, each blue vertex has neighbors in at most one white cluster (see~\Cref{claim:red_deg_one}).
Intuitively, RR 1 ensures that small clusters are removed, RR 1 and RR 2 ensure that there are no isolated white clusters, and RR 3 reduces the size of each white cluster to at most~$|B| + 1$.
}
\end{figure}
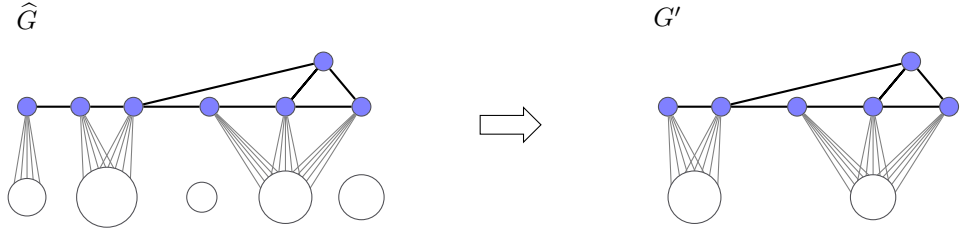
  
\newcommand\proofTwins{
Assume that $u$ and $v$ are adjacent in $G'$ and
assume towards a contradiction that there is a vertex~$w$ in~$G'$ which is adjacent to~$u$ in~$G'$ but not adjacent to~$v$ in~$G'$.
Since~$w$ and~$v$ are not adjacent in~$G'$, $G_R$ contains at most one of the arcs~$(w,v)$ or~$(v,w)$.
Assume without loss of generality that~$(w,v)$ is not an arc of~$G_R$.
Since~$u$ is adjacent to both~$v$ and~$w$ in~$G'$, $G_R$ contains the arcs~$(v,u)$ and~$(u,w)$.
Recall that both~$u$ and~$v$ are white vertices.
This implies that the arc-modification set~$M$ contains no arc incident with any of these two vertices.
Hence, $M$ contains none of the arcs of~$\{(v,u),(u,w),(v,w)\}$, which implies that~$G_R'= (V,A\Delta M)$ is not a transitive graph; a contradiction.
}

\newcommand{\proofRedDegOne}{
Assume towards a contradiction that there is a blue vertex~$w$ which is adjacent to two white vertices~$u$ and~$v$ in~$G'$, such that~$u$ and~$v$ are not part of the same white cluster.
	Since~$u$ and~$v$ are not part of the same white cluster, $u$ and~$v$ are not adjacent in~$G'$ due to~\Cref{claim:real twins}.
	This implies that~$G_R$ contains at most one of the arcs~$(u,v)$ or~$(v,u)$.
Assume without loss of generality that~$(u,v)$ is not an arc of~$G_R$.
Since~$w$ is adjacent to both~$u$ and~$v$ in~$G'$, $G_R$ contains the arcs~$(u,w)$ and~$(w,v)$.
By the fact that both~$u$ and~$v$ are white vertices, $M$ contains no arc of~$\{(u,w),(w,v),(u,v)\}$, which implies that~$G_R' = (V,A\Delta M)$ is not a transitive graph; a contradiction. 
}

\newcommand{\proofWhiteUniversal}{
Assume towards a contradiction that there is a vertex~$w$ in~$C$ that is not adjacent to~$u$.
Since~$C$ is a connected component in~$G'$, this implies that there is a vertex in~$C$ that has distance exactly two with~$u$.
We can assume without loss of generality that~$w$ is this vertex, that is, there is a vertex~$v$ in~$C$ that is adjacent to both~$u$ and~$w$.
	This implies that~$G_R$ contains at most one of the arcs~$(u,w)$ or~$(w,u)$.
Assume without loss of generality that~$(u,w)$ is not an arc of~$G_R$.
Since~$v$ is adjacent to both~$u$ and~$w$ in~$G'$, $G_R$ contains the arcs~$(u,v)$ and~$(v,w)$.
By the fact that~$u$ is a white vertex, $M$ does not contain the arc~$(u,w)$.
Moreover, since~$M$ contains no arc of~$\{(u,v),(v,w)\}$, we obtain that~$G_R' = (V,A\Delta M)$ is not a transitive graph; a contradiction. 
}

\begin{proof}
We first present a compression to~\CL.
Let~$\widehat{G}=(V,E)$ be an undirected graph that contains an edge~$\{u,v\}$ if and only if~$(u,v)$ and~$(v,u)$ are arcs of~$G_R$.
Due to~\Cref{lem:tcc-are-cliques}, $I$ is a yes-instance of~\OTCC if and only if~$(\widehat{G},k)$ is a yes-instance of~\CL.
Let~$W := V \setminus B$.
We call the vertices of~$B$~\emph{blue} and the vertices of~$W$~\emph{white}.
Note that~$G_R[W]$ is a transitive graph, since~$B$ is a transitivity modulator of~$G_R$.
Moreover, there exists an arc set~$M \subseteq B \times B$ such that~$G'_R = (V, A \Delta M)$ is transitive, since~$B$ is an inherent transitivity modulator of~$G_R$.
In the following, we present reduction rules to remove vertices from~$\widehat{G}$ to obtain an equivalent instance~$(G',k')$ of~\CL with $\Oh(|B|^2)$~vertices and where~$G'$ is an induced subgraph of~$\widehat{G}$.
The graphs~$\widehat{G}$ and~$G'$ are conceptually depicted in \Cref{fig:while-blue}. 

To obtain this smaller instance of~\CL, we initialize~$G'$ as a copy of~$\widehat{G}$ and~$k'$ as~$k$ and exhaustively applying three reduction rules.
Our first two reduction rules are the following:
\begin{itemize}
    \item RR 1: Remove a vertex~$v$ from~$G'$, if~$v$ has degree less than~$k'-1$ in~$G'$.
	\item RR 2: If a white vertex has at least~$k'-1$ white neighbors in~$G'$, output a constant size yes-instance. 
\end{itemize}   
Note that the first reduction rule is safe, since no vertex of degree less than~$k'-1$ can be part of a clique of size at least~$k'$.
Moreover, each connected component in~$G'$ has size at least~$k'$ after this reduction rule is exhaustively applied.
The safeness of the second reduction rule relies on the following observation.
\begin{claim}\label{claim:real twins}
If two white vertices~$u$ and~$v$ are adjacent in~$G'$, then they are real twins in~$G'$.
That is, $N_{G'}[u] = N_{G'}[v]$.
\end{claim}
\begin{claimproof} 
\proofTwins
\end{claimproof}

Note that this implies that each connected component in~$G'[W]$ is a clique of real twins in~$G'$.
We call each such connected component in~$G'[W]$ a~\emph{white cluster}.

Note that after exhaustive applications of the first two reduction rules, each white cluster has size at most~$k'-1$ and each connected component in~$G'$ has size at least~$k'$.
This implies that each connected component in~$G'$ contains at least one blue vertex.
Since~$G'$ contains at most $|B|$~blue vertices, this implies that~$G'$ has at most $|B|$~connected components.

In the following, we show that no blue vertex has neighbors in more than one white cluster.
This then implies that~$G'$ contains at most $|B| \cdot k'$~vertices.
In a final step, we then show how to reduce the value of~$k'$.
\begin{claim}
	\label{claim:red_deg_one}
	No blue vertex has neighbors in more than one white cluster.
\end{claim}
\begin{claimproof}
\proofRedDegOne
\end{claimproof}

As mentioned above, this implies that~$G'$ contains at most~$|B| \cdot k'$ vertices.
Next, we show how to reduce the size of the white clusters if~$k' > |B|$.
To this end, we introduce our last reduction rule:

\begin{itemize}
\item RR 3: If~$k' > |B| +1$, remove an arbitrary white vertex from each white cluster and reduce~$k'$ by 1. 
\end{itemize}

Note that RR 3 is safe:
If~$k' > |B|+1$, a clique of size~$k'$ in~$G'$ has to contain at least two white vertices, since~$G'$ contains at most~$|B|$ blue vertices.
Since no clique in~$G'$ can contain vertices of different white clusters, we reduce the size of a maximal clique of size at least~$k'$ in~$G'$ by exactly one, when removing one vertex of each white cluster.

Hence, after all reduction rules are applied exhaustively, the resulting instance~$(G',k')$ of~\CL contains at most~$|B|$ blue vertices and at most~$|B|$ white clusters.
Each such white cluster hast size at most~$|B|+1$.
This implies that the resulting graph~$G'$ contains $\Oh(|B|^2)$~vertices and $\Oh(|B|^3)$~edges, since each vertex has degree $\Oh(|B|)$.

Based on known polynomial-time reduction~\cite{BF03}, we can compute for an instance~$(G^*,k^*)$ of~\CL, an equivalent instance~$(\mathcal{G}^*,k^*)$ of~\OTCC, where~$\mathcal{G}^*$ is a proper temporal graph and has $\Oh(n+m)$~vertices and edges, where~$n$ and~$m$ denote the number of vertices and the number of edges of~$G^*$, respectively.
Since~$G'$ has $\Oh(|B|^2)$~vertices and $\Oh(|B|^3)$~edges, this implies that we can obtain an equivalent instance of~\OTCC of total size~$\Oh(|B|^3)$ in polynomial time.
By the fact that~$\mathcal{G}^*$ is a proper temporal graph, this works for all problem versions of~\OTCC.
This implies the stated kernelization result.
\end{proof}

Based on the fact that the set~$B$ of endpoints of any arc-modification set~$M$ towards a transitive graph is an inherent transitivity modulator of size at most~$2\cdot |M|$, this implies the following for kernelization algorithms with respect to arc-modification sets towards a transitive graph.

\begin{theorem}\label{Theorem kernel if set given}
Let~$I=(\mathcal{G},k)$ be an instance of~\OTCC and let~$G_R = (V,A)$ be the reachability graph of~$\mathcal{G}$.
Moreover, let~$M \subseteq V \times V$ be a set of arcs such that~$G'_R = (V, A \Delta M)$ is transitive.
Then, for each version of~\OTCC, one can compute in polynomial time an equivalent instance of total size~$\Oh(|M|^3)$.
\end{theorem}

Moreover, if the arc-modification set~$M$ only adds arcs to the reachability graph, we can obtain the further even better kernelization result.

\begin{lemma}\iflong\else[$\star$]\fi
Let~$I=(\mathcal{G},k)$ be an instance of~\OTCC and let~$G_R = (V,A)$ be the reachability graph of~$\mathcal{G}$.
Moreover, let~$M \subseteq V \times V$ be a set with~$A\cap M = \emptyset$ of arcs such that~$G'_R = (V, A \Delta M)$ is transitive.
Then, for each version of~\OTCC, one can compute in polynomial time an equivalent instance of total size~$\Oh(|M|^2)$. 
\end{lemma}
\iflong
\begin{proof}
First, we perform the compression to~\CL described in~\Cref{kernel if set given}.
Let~$B$ be the endpoints of the arcs of~$M$ and let~$W := V \setminus B$.
Note that~$B$ is a inherent transitivity modulator of~$G_R$ of size at most~$2\cdot |M|$.

Again, we call the vertices of~$B$~\emph{blue} and the vertices of~$W$~\emph{white}.
Let~$(G',k')$ be the instance of~\CL resulting after exhaustively applying the reduction rules RR 1, RR 2, and RR 3 to the instance~$(\widehat{G},k)$ (see \Cref{kernel if set given}).
Recall that~$G'$ contains at most~$|B|$ blue vertices and at most~$|B|$ white clusters.
Moreover, each white cluster has size at most~$|B| + 1$, no vertices of different white clusters are adjacent, and each blue vertex has neighbors in at most one white cluster.

In the following, we show how to improve the kernel-size, knowing that~$M$ contains no arc of~$G_R$.
To this end, we observe that each white vertex is adjacent to each vertex of its connected component in~$G'$.
\begin{claim}\iflong\else[$\star$]\fi\label{white are universal}
Let~$u$ be a white vertex and let~$C$ be the connected component of~$u$ in~$G'$.
Then, $u$ is adjacent to each other vertex of~$C$.
\end{claim}
\iflong
\begin{claimproof}
\proofWhiteUniversal
\end{claimproof}
\fi
Note that this further implies that if two blue vertices are adjacent in~$G'$, then both of these vertices have the same white neighbors.
Based on this observation, we now show how to obtain an equivalent instance~$(G'',k'')$ of~\CL, where~$G''$ has only $\Oh(|B|)$~vertices.

Let~$C_1, \dots, C_\ell$ be the connected components of~$G'$ and for each~$i\in [1,\ell]$, let~$\alpha_i$ denote the number of white vertices of~$C_i$, that is, $\alpha_i := |C_i\cap W|$.
We obtain the graph~$G''$ as follows:
Initially, we set~$G'' := G'[B]$.
Afterwards, we add a clique~$W^*$ of size~$\max_{i\in [1,\ell]} \alpha_i$ to~$G''$.
We fix some arbitrary order on the vertices of~$W^*$ and add for each~$i\in [1,\ell]$ an edge between each blue vertex of~$C_i$ and each of the first~$\alpha_i$ vertices of~$W^*$ to~$G''$.
Finally, we set~$k'':= k'$.
Note that~$G''$ contains $\Oh(|B|)$ vertices and $\Oh(|B|^2)$~edges, since the size of~$W^*$ is the size of the largest white cluster in~$G'$, which is~$\Oh(|B|)$.
Based on the above mentioned reduction to~\OTCC and the fact that~$|B| \leq 2\cdot |M|$, we can thus obtain an equivalent instance of~\OTCC with a proper temporal graph of total size~$\Oh(|M|^2)$.

To show the correctness of this algorithm it remains to show that~$G'$ contains a clique of size~$k'$ if and only if~$G''$ contains a clique of size~$k'$.

$(\Rightarrow)$
Let~$S$ be a clique of size~$k'$ in~$G'$.
Moreover, let~$C_i$ be the connected component of~$G'$ that contains~$S$.
By construction, the first~$\alpha_i = |C_i \cap W|$ vertices of~$W^*$ are adjacent to each blue vertex of~$S$.
Since~$S$ is a clique of size~$k'$ in~$G'$, this implies that~$(S\cap B) \cup W^*_{\leq \alpha_i}$ is a clique of size at least~$k'$ in~$G''$, where $W^*_{\leq \alpha_i}$ denotes the first~$\alpha_i$ vertices of~$W^*$.

$(\Leftarrow)$
Let~$S$ be a clique of size~$k'$ in~$G''$.
Since the second reduction rule is exhaustively applied, $k'$ exceeds the size of~$W^*$.
Consequently, $S$ contains at least one blue vertex.
Since the adjacency between blue vertices is the same in both~$G'$ and~$G''$, all white vertices of~$S$ belong to the same connected component of~$G'$, that is, there is some~$i\in [1,\ell]$, such that~$S\cap B \subseteq C_i$.
Since each vertex of~$C_i \cap B$ is adjacent only to~$\alpha_i = |C_i|$ vertices of~$W^*$, $(S\cap B) \cup (C_i \cap W)$ is a clique of size at least~$k'$ in~$G'$.
\end{proof}
\fi 

Hence, if we are given an arc-modification set~$M$ of size~$\tae$, we can compute a polynomial kernel.
Unfortunately, finding a minimum-size arc-modification set of a given directed graph is~\NP-hard~\cite{WKNU12} and no polynomial-factor approximations are known that run in polynomial time.
Hence, we cannot derive a polynomial kernel for the parameter~$\tae$.
Positively, if we only consider arc-additions, we can compute the transitive closure of a given directed graph in polynomial time.
This implies that we can find a minimum-size arc-modification set towards a transitive reachability graph in polynomial time among all such sets that only add arcs to to reachability graph.
Consequently, we derive the following.
\begin{corollary}
\OTCC admits a kernel of size~$\Oh(\taa^2)$, where~$\taa$ denotes the minimum number of necessary arc-additions to make the respective reachability graph transitive.
\end{corollary}

\section{Limits of these Parametrizations for Closed TCCs}
\label{sec:limits}
So far, the temporal paths that realize the reachability between two vertices in a tcc could lie outside of the temporal connected component. 
If we impose the restriction that those temporal paths must be contained in the tcc, the problem of finding a large tcc becomes \NP-hard even when the reachability graph is missing only a single arc to become a complete bidirectional clique:
In other words, the problem becomes~\NP-hard even if~$\tvd = \tae = 1$.
On general temporal graphs, all versions of~\CTCC are known to be~\NP-hard~\cite{Cast18,CLMS23}.

\begin{theorem}\label{ctcc is hard}
	For each version of~\CTCC, there is a polynomial time self-reduction that transforms an instance~$(\mathcal{G},k)$ with~$k> 4$ of that version of~\CTCC into an equivalent instance~$(\mathcal{G}',k)$, such that the reachability graph of~$\mathcal{G}'$ is missing only a single arc to be a complete bidirectional clique. 
\end{theorem}
\newcommand{\proofNearlyComplete}{
First, we show that~$(x_3,x_1)$ is not an arc of~$G_R'$.
By construction of~$\mg'$, (i)~no edge of~$\mg'$ that is incident with~$x_1$ exists in any time step larger than~$L + n + 2$, and (ii)~no edge of~$\mg'$ that is incident with~$x_3$ exists in any time step smaller than~$L + n + 3$.
This implies that no temporal path in~$\mg'$ that starts in~$x_3$ can reach~$x_1$.
Hence, $(x_3,x_1)$ is not an arc of~$G'_R$.

Next, we show that~$G_R'$ contains all other possible arcs.
To this end, we present strict temporal paths in~$\mg'$ that guarantee the existence of these arcs in~$G_R'$.
Let~$u$ and~$v$ be vertices of~$V$.
Consider the temporal path~$P$ that starts in vertex~$u$ and traverses
\begin{itemize}
\item the edge~$\{u,u'\}$ in time step~$L+1$,
\item the edge~$\{u',x_1\}$ in time step~$L+1+u$,
\item the edge~$\{x_1,x_2\}$ in time step~$L+n+2$,
\item the edge~$\{x_2,x_3\}$ in time step~$L+n+3$,
\item the edge~$\{x_3,v'\}$ in time step~$L+n+3+v$, and
\item the edge~$\{v',v\}$ in time step~$L+2 n+4$.
\end{itemize} 
Since each vertex of~$V$ is a natural number of~$[1,n]$, the times steps in which these edges are traversed by~$P$ are strictly increasing, which implies that~$P$ is a strict temporal path of~$\mg'$.
Note that each suffix and each prefix of~$P$ is also a strict temporal path in~$\mg'$.
Hence, this implies that~$G_R'$ contains all the arcs~$\{(u,v'),(u,v),(u',v'),(u',v)\} \cup \{(u',x_1),(u',x_1),(x_1,v'),(x_1,v)\mid i\in \{1,2,3\}\} \cup \{(x_1,x_3)\}$.
Moreover, since~$\{x_1,x_2\}$ and~$\{x_2,x_3\}$ are edges in~$\mg'$, $G_R'$ also contains the arcs~$\{(x_1,x_2),(x_2,x_1),(x_2,x_3),(x_3,x_2)\}$.
This implies that~$(x_3,x_1)$ is the unique arc that is missing in~$G_R'$.
}
\newcommand{\proofXNotInS}{
Since~$k>4$, $S$ contains at least one vertex~$v^*$ which is not contained in~$\{x_1,x_2,x_3\}$.
Note that due to~\Cref{nearly complete}, $S$ contains at most one of~$x_1$ and~$x_3$.
We distinguish two cases.

\textbf{Case 1:} $x_1$ is not contained in~$S$\textbf{.}
We show that there is no temporal path from~$v^*$ to~$x_2$ in~$\mg'$ that does not visit~$x_1$.
This then implies that~$x_2$ is not contained in~$S$.
Recall that~$\{x_2,x_3\}$ is the unique edge incident with~$x_2$ in~$\mg'$ that does not have~$x_1$ as an endpoint.
Moreover, recall that this edge can only be traversed at time step~$L+n+3$.
By the fact that each other edge incident with~$x_3$ exists only in time steps strictly greater than~$L+n+3$, there is no temporal path from~$v^*$ to~$x_2$ in~$\mg'$ that does not visit~$x_1$.

\textbf{Case 2:} $x_3$ is not contained in~$S$\textbf{.}
We show that there is no temporal path from~$x_2$ to~$v^*$ in~$\mg'$ that does not visit~$x_3$.
This then implies that~$x_2$ is not contained in~$S$.
Recall that~$\{x_1,x_2\}$ is the unique edge incident with~$x_2$ in~$\mg'$ that does not have~$x_3$ as an endpoint.
Moreover, recall that this edge can only be traversed at time step~$L+n+2$.
By the fact that each other edge incident with~$x_1$ exists only in time steps strictly smaller than~$L+n+2$, there is no temporal path from~$x_2$ to~$v^*$ in~$\mg'$ that does not visit~$x_3$.
}
\newcommand{\proofNoAuxVerts}{
Assume towards a contradiction that there is a temporal path~$P$ from~$v'$ to~$u'$ in~$\mg'$ that does not visit~$x_2$.
Since each edge incident with~$v'$ in~$\mg'$ only exists in time steps strictly larger than~$L$, $P$ cannot visit any vertex~$w\in V$.
This is due to the fact that~$\{w,w'\}$ is the only edge incident with~$w$ in~$\mg'$ that exists in any time step strictly larger than~$L$.

Since~$v'$ is only incident with the three edges~$\{v',v\}$, $\{v',x_1\}$, and $\{v',x_3\}$ in~$\mg'$, the above implies that~$P$ first traverses~$\{v',x_1\}$ or~$\{v',x_3\}$.
Similarly, the last edge traversed by~$P$ is either~$\{x_1,u'\}$ or~$\{x_3,u'\}$.
By construction, the edge~$\{v',x_3\}$ only exists at time step~$L + 2n + 3 + v$ which is strictly larger than the time steps in which~$\{x_1,u'\}$ or~$\{x_3,u'\}$ exists, namely time step~$L + 1 + u$ and time step~$L + 2n + 3 + u$, respectively.
This implies that~$P$ does not traverse the edge~$\{v',x_3\}$.
Similarly, since the edge~$\{v',x_1\}$ only exists at time step~$L + 1 + v$ which is strictly larger than the time steps in which~$\{x_1,u'\}$ exists, $P$ firstly traverses the edge~$\{v',x_1\}$ and lastly traverses the edge~$\{x_3,u'\}$.

Since~$P$ does not visit~$x_2$, for some vertex~$w\in V$, $P$ secondly traverses the edge~$\{x_1,w'\}$, and thirdly traverses the edge~$\{w',x_3\}$.
Moreover, since~$P$ is a temporal path in~$\mg'$, the edge~$\{x_1,w'\}$ exists in a time step at least as large as the time step in which~$\{v',x_1\}$ exists.
By definition of~$\mg'$, this implies that~$w > v$.
Since~$v > u$, this implies that the edge~$\{w',x_3\}$ exists only at time step~$L + n + 3 + w$, which is strictly larger than the time step in which the edge~$\{x_3, u'\}$ exists.
Hence,~$P$ is not a temporal path in~$\mg'$; a contradiction.
}
\newcommand{\proofOnlyOrigin}{
Recall that~$S$ contains at most one of~$x_1$ and~$x_3$.
Due to~\Cref{x2 not in s}, $S$ does not contain~$x_2$.
By~\Cref{no aux vertes reach}, this implies that~$S$ contains at most one vertex of~$\{w'\mid w\in V\}$.
Since~$k> 4$, we have that~$S$ contains at least two vertices~$u$ and~$v$ of~$V$.

First, we show that~$S$ contains neither vertex~$x_1$ nor vertex~$x_3$.
Assume towards a contradiction that~$S$ contains~$x_1$ or~$x_3$.
We distinguish two cases.

\textbf{Case 1:} $S$ contains~$x_1$\textbf{.}
Note that then~$x_3$ is not contained in~$S$.
Since~$S$ contains~$x_1$, $u$, and~$v$, for each~$w\in \{u,v\}$, there is a temporal path~$P_w$ in~$\mg'$ from~$x_1$ to~$w$ that only visits vertices of~$S$.
In particular, $P_w$ does not visit~$x_3$.
By construction, each edge incident with~$x_1$ in~$\mg'$ only exists in time steps strictly larger than~$L$, hence, $P_w$ does not traverse any edge between vertices of~$V$.
This implies that~$P_w$ is the unique temporal path that firstly traverses the edge~$\{x_1,w'\}$ and secondly traverses the edge~$\{w',w\}$.
This implies that~$u'$ and~$v'$ are both contained in~$S$.
A contradiction to the fact that~$S$ contains at most one vertex of~$\{w'\mid w\in V\}$.

\textbf{Case 2:} $S$ contains~$x_3$\textbf{.}
This case follows by the same arguments as the first case, when replacing~$x_1$ by~$x_3$ and vice versa.

Hence, we obtain that~$S$ contains none of the vertices of~$\{x_1,x_2,x_3\}$.
Hence, it remains to show that~$S$ contains no vertex of~$\{w'\mid w\in V\}$.
Assume towards a contradiction that there is a vertex~$w\in V$, such that~$w'$ is contained in~$S$.
Since~$u$ and~$v$ are distinct, assume without loss of generality that~$w \neq u$.
Since~$S$ contains both~$w'$ and~$u$, there is a temporal path~$P$ in~$\mg'$ from~$w'$ to~$u$ that visits neither~$x_1$ nor~$x_3$.
Recall that~$w'$ is only incident with the three edges~$\{w',w\}$, $\{w',x_1\}$, and $\{w',x_3\}$ in~$\mg'$.
Since~$P$ visits neither~$x_1$ nor~$x_3$, the first edge traversed by~$P$ is the edge~$\{w',w\}$ at a time step strictly larger than~$L$.
By construction~$\{w',w\}$ is the only edge incident with~$w$ in~$\mg'$ that exists in any time step strictly larger than~$L$, hence, $P$ cannot reach vertex~$u$, since~$u \neq w$.
A contradiction.

Concluding, $S$ contains only vertices of~$V$.
}
\begin{proof}
Let~$I:=(\mathcal{G},k)$ be an instance of~\CTCC with underlying graph~$G = (V,E)$ and let~$L$ be the lifetime of~$\mg$.
Moreover, assume for simplicity that the vertices of~$V$ are exactly the natural numbers from~$[1,n]$ with~$n:= |V|$.
To obtain an equivalent instance~$I':=(\mathcal{G}',k)$ of~\CTCC, we extends~$\mg$ as follows:
We initialize~$\mg'$ as a copy of~$\mg$ and for each vertex~$v\in V$, we add  a new vertex~$v'$ to~$\mg'$.
Additionally, we add three vertices~$x_1$, $x_2$, and~$x_3$ to~$\mg'$.
Furthermore, we append~$2n+4$ empty snapshots to the end of~$\mg'$ and add the following edges to~$\mg'$:
For each vertex~$v\in V$, we add the edge~$\{v,v'\}$ to time steps~$L+1$ and~$L + 2n + 4$, the edge~$\{v', x_1\}$ to time step~$L + 1 + v$, and the edge~$\{v',x_3\}$ to time step~$L + n + 3 + v$.
Finally, we add the edge~$\{x_1,x_2\}$ to time step~$L + n + 2$ and the edge~$\{x_2,x_3\}$ to time step~$L + n + 3$.
This completes the construction of~$\mg'$.
An illustration of the additional vertices and edges is given in~\Cref{figure}.

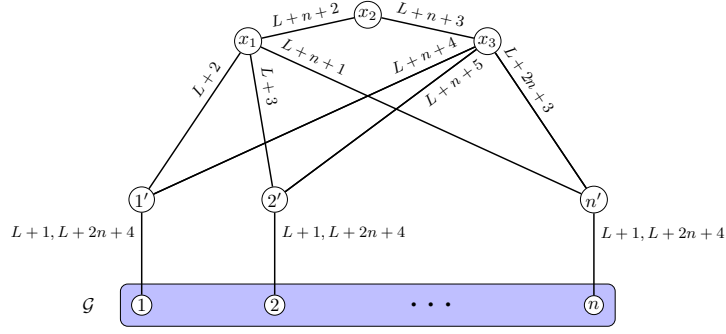
\begin{figure}
\centering
\scalebox{0.7}{
\begin{tikzpicture}
\tikzstyle{knoten}=[circle,fill=white,draw=black,minimum size=5pt,inner sep=1pt]

\node[knoten] (x1) at (0,0) {$x_1$};
\node[knoten] (x2) at ($(x1) + (2.25,.5)$) {$x_2$};
\node[knoten] (x3) at ($(x1) + (4.5,0)$) {$x_3$};

\node[knoten] (v1) at ($(x1) + (-2,-3)$) {$1'$};
\node[knoten] (v2) at ($(v1) + (2.5,0)$) {$2'$};
\node[knoten] (vn) at ($(v2) + (6,0)$) {$n'$};

\draw[thick] (x1) -- (x2) node [midway,above, sloped] {\footnotesize $L + n + 2$};
\draw[thick] (x2) -- (x3) node [midway,above, sloped] {\footnotesize $L + n + 3$};

\foreach \x [count=\y] in {1,2,n} {
    \draw[thick] (x3) to (v\x);
    
    \node[knoten] (w\x) at ($(v\x) + (0,-2)$) {$\x$};
    
    \ifthenelse{\y=1}{
    \draw[thick] (w\x) -- (v\x) node [left, sloped, near end, rotate=270] {\footnotesize $L+1, L + 2n + 4$};
    }{
    \draw[thick] (w\x) -- (v\x) node [right, sloped, near end, rotate=270] {\footnotesize $L+1, L + 2n + 4$};
    }   
}

    \draw[thick] (x1) -- (v1) node [above,sloped, near start, rotate=0] {\footnotesize $L + 2$};
    \draw[thick] (x1) -- (v2) node [above,sloped, near start, rotate=0] {\footnotesize $L + 3$};
    \draw[thick] (x1) -- (vn) node [above,sloped, very near start, rotate=0] {\footnotesize $~~~L + n + 1$};

    \draw[thick] (x3) -- (v1) node [above,sloped, very near start, rotate=0] {\footnotesize $L + n + 4~~$};
    \draw[thick] (x3) -- (v2) node [below,sloped, very near start, rotate=0] {\footnotesize $L + n + 5~~~$};
    \draw[thick] (x3) -- (vn) node [above,sloped, near start, rotate=0] {\footnotesize $L + 2n + 3$};

\node (mg) at ($(w1) + (-1,0)$) {$\mg$};

\node (dots) at ($.5*(w2) + .5*(wn)$) {\huge $\cdots$};

\begin{pgfonlayer}{background}

\draw [rounded corners, fill=blue!25] ($(w1) + (-.4, -.4)$) rectangle ($(wn) + (.4, .4)$) {};

\end{pgfonlayer}

\end{tikzpicture}
}

\caption{An illustration of the additional vertices and edges that are added to~$\mg$ in the reduction of~\Cref{ctcc is hard}.
Here, $L$ denotes the lifetime of~$\mg$ and the labels on the edges indicate in which snapshots the respective edges exist in the constructed temporal graph.}
\label{figure}
\end{figure}

Before we show the equivalence between the two instances of~\CTCC, we first show that the reachability graph~$G_R'$ of~$\mg'$ only misses a single arc to be a bidirectional clique.

\begin{claim}\label{nearly complete}
The arc~$(x_3,x_1)$ is the only arc that is missing in~$G_R'$.
\end{claim}
\begin{claimproof}
\proofNearlyComplete
\end{claimproof}

Next, we show that~$I$ is a yes-instance of~\CTCC if and only if~$I'$ is a yes-instance of~\CTCC.

$(\Rightarrow)$
This direction follows directly by the fact that~$\mg'$ is obtained by extending~$\mg$.
Hence, a closed tcc~$S$ of size~$k$ in~$\mg$ is also a closed tcc in~$\mg'$.

$(\Leftarrow)$
Let~$S$ be a closed tcc of size~$k$ in~$\mg'$.
Recall that~$k>4$.
We show that~$S$ only contains vertices of~$V$.
To this end, we first show that~$S$ does not contain~$x_2$.

\begin{claim}\iflong\else[$\star$]\fi\label{x2 not in s}
The set~$S$ does not contain~$x_2$.
\end{claim}

\iflong
\begin{claimproof}
\proofXNotInS
\end{claimproof}
\fi

Next, we show that the vertex~$x_2$ is required to have pairwise temporal paths between distinct vertices of~$\{w'\mid w\in V\}$ in~$\mg'$.

\begin{claim}\iflong\else[$\star$]\fi\label{no aux vertes reach}
Let~$u$ and~$v$ be distinct vertices of~$V$ with~$u < v$.
Then, each temporal path from~$v'$ to~$u'$ in~$\mg'$ visits~$x_2$. 
\end{claim}
\iflong
\begin{claimproof}
\proofNoAuxVerts
\end{claimproof}
\fi

As a consequence,~$S$ contains at most one vertex of~$\{w'\mid w\in V\}$, since~$S$ is a closed tcc that does not contain vertex~$x_2$. 
Based on the above two claims, we now show that~$S$ contains only vertices of~$V$.

\begin{claim}\iflong\else[$\star$]\fi\label{only original vertices}
Only vertices of~$V$ are contained in~$S$.
\end{claim}
\iflong
\begin{claimproof}
\proofOnlyOrigin
\end{claimproof}
\fi

Since no edge between any two vertices of~$V$ was added while constructing~$\mg'$ from~$\mg$, each temporal path in~$\mg'$ that visits only vertices of~$V$ is also a temporal path in~$\mg$.
Together with~\Cref{only original vertices}, this implies that~$S$ is a closed tcc in~$\mg$, which implies that~$I$ is a yes-instance of~\CTCC.
\end{proof}

Recall that all versions of \CTCC are \NP-hard~\cite{Cast18,CLMS23}.
Moreover the strict undirected version of \CTCC is \W-hard when parameterized by~$k$~\cite{Cast18} and both directed versions of \CTCC are \W-hard when parameterized by~$k$~\cite{CLMS23}.
Together with~\Cref{ctcc is hard}, this implies the following intractability results for~\CTCC.

\begin{theorem}
All versions of \CTCC are \NP-hard even if $\tvd = \tae = 1$.
	More precisely, this hardness holds on instances where the reachability graph is missing only a single arc to be a complete bidirectional clique. 
	Excluding the undirected non-strict version of~\CTCC, all versions of \CTCC are \W-hard when parameterized by~$k$ under these restrictions
\end{theorem}

\section{Conclusion}
\label{sec:conclusion}
We introduced two new parameters $\tvd$ and $\tae$ that capture how far the reachability graph of a given temporal graph is from being transitive. 
We demonstrated their applicability when the goal is to find open tccs in a temporal graph, presenting \FPT-algorithms for each parameter individually, and a polynomial kernel with respect to~$\tae$, assuming that the corresponding arc-modification set of size $\tae$ is given.
Computing such a set is \NP-hard in general directed graphs~\cite{WKNU12}. 
An interesting question, also formulated in that paper, is whether this parameter is at least approximable to within a polynomial factor of $\tae$. 
If so, our result implies a polynomial kernel for~\OTCC when parameterized by~$\tae$.
Alternatively, the existence of a proper polynomial kernel for~\OTCC when parameterized by~$\tae$ could also be shown by finding an approximation for a minimum-size inherit transitivity modulator due to~\Cref{kernel if set given} and the fact that the size of a minimum-size inherit transitivity modulator never exceeds~$2\cdot \tae$.

Another natural question is to identify what are the other (temporal) reachability problems for which our transitivity parameters could be useful. For instance, consider a variant of the \OTCC problem where we search for \emph{$d$-tccs}, that is, tccs such that the fastest temporal path between the vertices has duration at most~$d$.
It is plausible that our positive results carry over to this version when applied to the~\emph{$d$-reachability graph}, i.e., the graph whose arcs represent temporal paths of duration at most~$d$.

Regarding \CTCC, our intractability results show that neither \tvd nor \tae suffice to make this problem tractable. However, our results do not preclude the existence of an \FPT-algorithm in the case that the arc modification operations are restricted to deletion only, which remains to be investigated. Nonetheless, we still believe that transitivity is a key aspect of the problem. The problem with \CTCC is that the reachability graph itself does not encode whether the paths responsible for reachability travel through internal or external vertices.

We would like to initiate the idea of considering further transitivity parameters based on modifications of the temporal graph itself, not only of the reachability graph. In particular, could \FPT-algorithms for such parameters be achieved for reachability problems such as \OTCC with similar performance as our \FPT-algorithms for \tvd and \tae, and could these parameters make \CTCC tractable as well? And if so, how difficult is the computation of such parameters?

\newpage

\bibliography{bib}

\end{document}